\documentclass{llncs}

\usepackage{llncsdoc}
\usepackage[utf8]{inputenc}
\usepackage{amsfonts}
\usepackage{mathtools}
\usepackage{multirow}
\usepackage{units}
\usepackage{algpseudocode}
\usepackage{algorithm}
\usepackage{enumitem}
\usepackage{bbm}
\usepackage{diagbox}
\usepackage{siunitx}
\usepackage{etoolbox}

\DeclareMathOperator*{\argmax}{arg\,max}

\title{K-clique-graphs for Dense Subgraph Discovery}
\author{Giannis Nikolentzos\inst{1,2} \and Polykarpos Meladianos\inst{1,2} \and Yannis Stavrakas\inst{3} \and Michalis Vazirgiannis\inst{1,2}}
\institute{Lix, École Polytechnique, France\\
\and
Athens University of Economics and Business, Greece\\
\and
Institute for the Management of Information Systems RC ``Athena'', Greece\\
\email{\{nikolentzos,pmeladianos,mvazirg\}@aueb.gr, yannis@imis.athena-innovation.gr}}

\begin{document}

\maketitle            

\begin{abstract}
Finding dense subgraphs in a graph is a fundamental graph mining task, with applications in several fields.
Algorithms for identifying dense subgraphs are used in biology, in finance, in spam detection, etc. Standard formulations of this problem such as the problem of finding the maximum clique of a graph are hard to solve.
However, some tractable formulations of the problem have also been proposed, focusing mainly on optimizing some density function, such as the degree density and the triangle density.
However, maximization of degree density usually leads to large subgraphs with small density, while maximization of triangle density does not necessarily lead to subgraphs that are close to being cliques.

In this paper, we introduce the $k$-clique-graph densest subgraph problem, $k \geq 3$, a novel formulation for the discovery of dense subgraphs.
Given an input graph, its $k$-clique-graph is a new graph created from the input graph where each vertex of the new graph corresponds to a $k$-clique of the input graph and two vertices are connected with an edge if they share a common $k - 1$-clique.
We define a simple density function, the $k$-clique-graph density, which gives compact and at the same time dense subgraphs, and we project its resulting subgraphs back to the input graph.
In this paper, we focus on the triangle-graph densest subgraph problem obtained for $k = 3$.
To optimize the proposed function, we present an efficient greedy approximation algorithm that scales well to larger graphs.

We evaluate the proposed algorithm on real datasets and compare it with other algorithms in terms of the size and the density of the extracted subgraphs.
The results verify the ability of the proposed algorithm in finding high-quality subgraphs in terms of size and density.
Finally, we apply the proposed method to the important problem of keyword extraction from textual documents.
\end{abstract}

\section{Introduction}\label{sec:introduction}
\noindent
In recent years, graph-based representations have become extremely popular for modelling real-world data.
Some examples of data represented as graphs include social networks, protein or gene regulation networks and textual documents.
The problem of extracting dense subgraphs from such graphs has received a lot of attention due to its potential applications in many fields.
Specifically, in the web graph, dense subgraphs may correspond to link spam \cite{gibson2005discovering} and hence, they can be used for spam detection.
In bioinformatics, they are used for finding molecular complexes in protein-protein interaction networks \cite{bader2003automated} and for discovering motifs in genomic DNA \cite{fratkin2006motifcut}.
In the field of finance, they are used for discovering migration motifs in financial markets \cite{du2009migration}.
Other applications include graph compression \cite{buehrer2008scalable}, graph visualization \cite{alvarez2005large}, real-time identification of important stories in Twitter \cite{angel2014dense} and community detection \cite{chen2012dense}.

Given an undirected, unweighted graph $G = (V,E)$, we will denote $|V| = n$ the number of vertices and $|E| = m$ the number of edges.
Given a subset of vertices $S \subseteq V$, let $E(S)$ be the set of edges that have both end-points in $S$.
Hence, $G(S) = (S, E(S))$ is the subgraph induced by $S$.
The \textit{density} of the set $S$ is $\delta(S) = |E(S)|/\binom{|S|}{2}$, the number of edges in $S$ over the total possible edges.
Finding the set $S$ that maximizes $\delta$ is not a meaningful problem, as density $\delta$ does not take into account the size of the subgraph.
For example, a subgraph consisting of two vertices connected with an edge has higher density $\delta$ than a subgraph consisting of $100$ vertices and all but one edge between them.
However, clearly, we would prefer the latter subgraph from the former even if it achieves a lower value of density $\delta$.
Typically, the problem of dense subgraph discovery asks for a set of vertices $S$ which is large and which has high density.
Several different functions have been proposed in the literature that aim to solve this problem.
Some of these functions can be optimized in polynomial time, however, most of these formulations of extracting dense subgraphs are NP-hard and also hard to approximate.

Recently, there was a growing interest in the extraction of subgraphs whose vertices are highly connected to each other \cite{tsourakakis2013denser,balalau2015finding,tsourakakis2015k}.
However, existing methods do not always find subgraphs with high density $\delta$.
Instead, they prefer subgraphs with many vertices even if their density $\delta$ is not very high.
In many cases, we are interested in discovering sets of vertices where there is an edge between almost all their pairs.
In this paper, we introduce a new formulation for extracting dense subgraphs.
We define a new family of functions for measuring the density of a subgraph and we provide exact and approximate algorithms that allow the extraction of large subgraphs with high density $\delta$ by maximizing these functions. Our contributions are fourfold:
\begin{enumerate}[label=(\roman*),leftmargin=.4cm]
  \item \textbf{New formulation:} We introduce the \textit{$k$-clique-graph densest subgraph} ($k$-clique-GDS) problem, a new formulation for finding large subgraphs with high density $\delta$.
  Given a value for $k$, we create a graph whose vertices correspond to $k$-cliques of the original graph and we draw edges between two $k$-cliques if they share a common $(k-1)$-clique.
  We then extract a dense subgraph from the new graph and we project the result back to the original graph.
  We focus on the special case obtained for $k = 3$ which we call the \textit{triangle-graph densest subgraph} (TGDS) problem.
  We define a new density function which is suited to the needs of our problem.
  \item \textbf{Approximation algorithm:} We propose an efficient greedy approximation algorithm for the TGDS problem which removes one vertex at each iteration.
  The algorithm achieves nearly-optimal results on real-world networks.
  \item \textbf{Experimental evaluation:} We evaluate our approximation algorithm on several real-world networks.
  We compare the obtained subgraphs with those outputted by state-of-the-art algorithms and we observe that the proposed algorithm extracts subgraphs of high quality.
  We also present an application of our problem to the task of keyword extraction from textual documents.
\end{enumerate}

\section{Related Work}\label{sec:related_work}
\noindent
In this Section, we review the related work published in the areas of \textit{Clique Finding}, \textit{Dense Subgraph Discovery} and \textit{Triangle Listing}.\\
\textbf{Clique Finding.}
A clique is a graph whose vertices are all connected to each other.
Hence, all cliques have density $\delta = 1$.
A \textit{maximum} clique of a graph is a clique, such that there is no clique with more vertices.
Finding the maximum clique of a graph is an NP-complete problem \cite{karp1972reducibility}.
The maximum clique problem is also hard to approximate.
More specifically, H{\aa}stad showed in \cite{hastad1996clique} that for any $\epsilon > 0$, there is no polynomial algorithm that approximates the maximum clique within a factor better than $\mathcal{O}(n^{1-\epsilon})$, unless NP has expected polynomial time algorithms.
Feige presented in \cite{feige2004approximating} a polynomial-time algorithm that approximates the maximum clique within a ratio of $\mathcal{O}(\nicefrac{n (\log \log n)^2}{(\log n)^3})$.
A \textit{maximal} clique is a clique that is not included in a larger clique.
The Bron–Kerbosch algorithm is a recursive backtracking procedure \cite{bron1973algorithm} that lists all maximal cliques in a graph in $\mathcal{O}(3^{n/3})$ time.\\
\textbf{Dense Subgraph Discovery.}
The problem of finding a dense subgraph given an input graph has been widely studied in the literature \cite{lee2010survey}.
As mentioned above, such a problem aims at finding a subset of vertices $S \subseteq V$ of an input graph $G$ that maximizes some notion of density.
Among all the functions for evaluating dense subgraphs, degree density has gained increased popularity.
The degree density of a set of vertices $S$ is defined as $d(S) = 2|E(S)|/|S|$.
The problem of finding the set of vertices that maximizes the degree density is known as the \textit{densest subgraph} (DS) problem.
The set of vertices $S \subseteq V$ that maximizes the degree density can be identified in polynomial time by solving a series of minimum-cut problems \cite{goldberg1984finding}.
Charikar showed in \cite{charikar2000greedy} that the DS problem can also be formulated as a linear programming (LP) problem.
In the same paper, the author proved that the greedy algorithm proposed by Asahiro et al. \cite{asahiro2000greedily} provides a $\frac{1}{2}$-approximation to the DS problem in linear time.

Some variations of the DS problem include the \textit{densest $k$-subgraph} (DkS), the \textit{densest at-least-$k$-subgraph} (DalkS) and the \textit{densest at-most-$k$-subgraph} (DamkS) problems.
These variations put restrictions on the size of the extracted subgraph.
The DkS identifies the subgraph with exactly $k$ vertices that maximizes the degree density and is known to be NP-complete \cite{asahiro2002complexity}.
Feige et al. provided in \cite{feige2001dense} an approximation algorithm with approximation ratio $\mathcal{O}(n^{\delta})$, where $\delta < 1/3$.
The DalkS and DamkS problems were introduced by Andersen and Chellapilla \cite{andersen2009finding}.
The first problem asks for the subgraph of highest degree density among all subgraphs with at least $k$ vertices and is known to be NP-hard \cite{khuller2009finding}, while the second problem asks for the subgraph of highest density among all subgraphs with at most $k$ vertices and is known to be NP-complete \cite{andersen2009finding}.

Tsourakakis introduced in \cite{tsourakakis2015k} the \textit{$k$-clique densest subgraph} ($k$-clique-DS) problem which generalizes the DS problem.
The $k$-clique-DS problem maximizes the average number of $k$-cliques induced by a set $S \subseteq V$ over all possible vertex subsets.
For $k = 3$, we obtain the so-called \textit{triangle densest subgraph} (TDS) problem which maximizes the triangle density defined as $d_{tr}(S) = t(S)/|S|$ where $t(S)$ is the number of triangles in $S$.
The author provides two polynomial-time algorithms that identify the exact set of vertices that maximizes the triangle density and a $\frac{1}{3}$-approximation algorithm which runs asymptotically faster than any of the exact algorithms.

There are several other recent algorithms that extract dense subgraphs by maximizing other notions of density \cite{sozio2010community,tsourakakis2013denser,wang2010triangulation}.
It is worthwhile mentioning Tsourakakis et al.'s work \cite{tsourakakis2013denser}.
The authors defined the \textit{optimal quasi-clique} (OQC) problem which finds the subset of vertices $S \subseteq V$ that maximizes the function $f_{\alpha}(S) = |E(S)| - \alpha \binom{|S|}{2}$ where $\alpha \in (0,1)$ is a constant.
The OQC problem is not polynomial-time solvable and the authors provided a greedy approximation algorithm that runs in linear time and a local-search heuristic.\\
\textbf{Triangle Listing.}
Given a graph $G$, the \textit{triangle listing} problem reports all the triangles in $G$.
The triangle listing problem has been extensively studied and a large number of algorithms has been proposed \cite{itai1978finding,chiba1985arboricity,schank2005finding}.
A listing algorithm requires at least one operation per triangle.
In the worst case, there are $n^3$ triangles in terms of the number of vertices and $m^{3/2}$ in terms of the number of edges.
Hence, in the worst case, it takes $m^{3/2}$ time just to report the triangles.
The above algorithms require $\mathcal{O}(m^{3/2})$ time to list the triangles and they are thus optimal in the worst case.
Recently, Bj{\"o}rklund et al. proposed output sensitive algorithms which run asymptotically faster when the number of triangles in the graph is small \cite{bjorklund2014listing}.

\section{Problem Definition}\label{sec:problem_definition}
\noindent
In this Section, we will introduce the \textit{$k$-clique-graph densest subgraph} ($k$-clique-GDS) problem, a novel formulation for finding dense subgraphs.
In the following, we will restrict ourselves to the case where $k=3$, that is to triangles.
At the end of the Section, we will describe how the proposed approach can be generalized to the case of $k$-cliques, $k > 3$.

The cornerstone of the proposed method is the transformation of the input graph $G = (V,E)$ into another graph $G' = (V',E')$.
The transformed graph $G'$ is a more abstract representation of $G$.
Specifically, it encodes information regarding the triangles of the input graph $G$ and the relationships between them.

As a preprocessing step before applying the transformation, we assign labels to the edges of the input graph $G$.
Given a set of labels $L$, $\ell : E \rightarrow L$ is a function that assigns labels to the edges of the graph.
Each edge is assigned a unique label.
Hence, the cardinality of the set $L$ is equal to that of set $E$, $|L|=|E|$.
We next proceed with the transformation of $G$ into $G'$.
The first step of the transformation procedure is to run a triangle listing algorithm.
There are several available triangle listing algorithms as described in Section~\ref{sec:related_work}.
Let $T(S)$ be the set of triangles extracted from $G$.
For each triangle $t \in T(G)$, we create a vertex in the new graph $G'$.
Therefore, each vertex represents one of the triangles extracted from $G$.
Pairs of triangles that share a common edge in $G$ are considered neighbors and are connected with an edge in $G'$.
In other words, each edge in $G'$ corresponds to a pair of triangles sharing the same edge.
The edges of $G'$ are also assigned labels.
Each edge in $G'$ is given the label of the edge that is shared between the two corresponding triangles in $G$.
For example, given a pair of triangles $t_1 = (v_1,v_2,v_3)$ and $t_2 = (v_1,v_2,v_4)$ where $t_1,t_2 \in T(G)$, these triangles have a common edge $e = (v_1,v_2)$ and the edge $e'$ that links them in $G'$ gets the same label as $e$, that is $\ell(e') = \ell(e)$.
A triangle has three edges, hence, although it can have any number of adjacent edges in $G'$, its labels come from a limited alphabet consisting of only three items (the labels of the three edges of the triangle in $G$).
We call the transformed graph $G'$ the \textit{triangle-graph} of $G$.
Algorithm~\ref{alg:transformation} 
\begin{algorithm}[t]
\caption{Construct triangle-graph}
\label{alg:transformation}
\begin{algorithmic}
\Require{graph $G = (V,E)$}
\Ensure{graph $G' = (V',E')$}
\State{1: Assign a unique label to each edge of the input graph $G$.}
\State{2: Extract all triangles in $G$ by running a triangle listing algorithm. Let $T(S)$ be the set of the extracted triangles.}
\State{3: Create a new empty graph $G'$.}
\State{4: For each triangle $t \in T(G)$ create a vertex in the $G'$.}
\State{5: Connect two vertices in $G'$ with an edge if the corresponding triangles in $G$ share a common edge.}
\State{6: Assign to the new edge the label of the edge that is shared between the two triangles.}
\State{7: Return $G'$.}
\end{algorithmic}
\end{algorithm}
describes the steps required to create $G'$ from $G$ and Figure~\ref{fig:transformation}
illustrates how a graph containing $4$ triangles is transformed into its triangle-graph.

After creating the triangle-graph $G'$, we can find a subset of vertices $S' \subset V'$ that correponds to a dense subgraph.
As mentioned earlier, each vertex $v \in S'$ represents a triangle $t$ of the input graph $G$.
Each triangle $t$ is a set of three vertices.
Intuitively, the union of the vertices of all the triangles that belong to the set $S'$ will form a dense subgraph of $G$.
To extract the set of vertices $S'$, we can define a density measure and optimize it.
A simple measure we can employ is the well-known degree density defined as $d(S') = 2|E(S')|/|S'|$.
However, the above function will not necessarily lead to subgraphs with high density.
Consider the two graphs shown in Figure~\ref{fig:comparison}.
\begin{figure}[t]
  \centering
  \includegraphics[width=.7\linewidth]{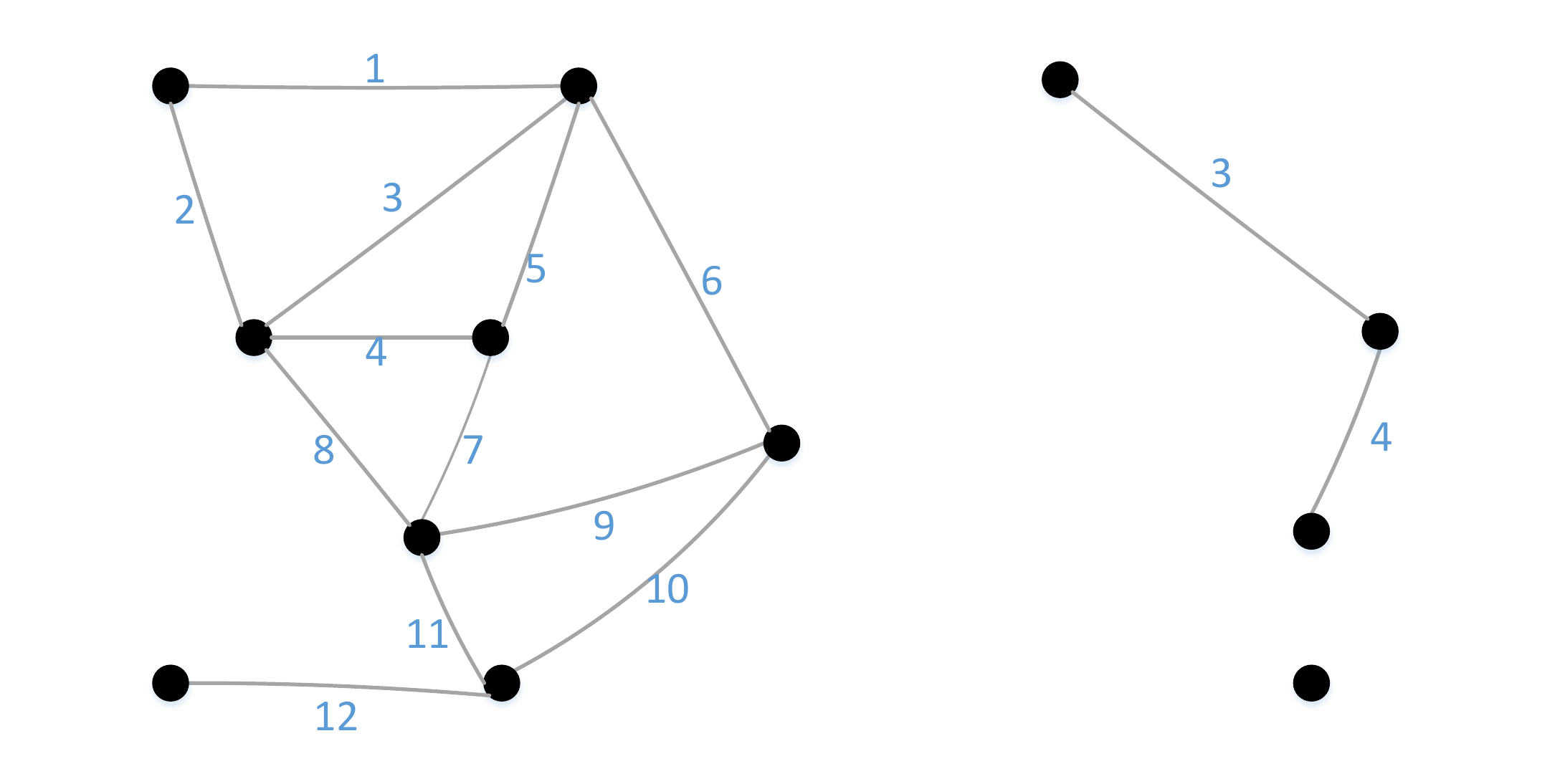}
  \caption{Example of an input graph (left)  and the triangle-graph (right) created from it. There are $4$ triangles in the input graph defined by the following triads of edges: ($1,2,3$), ($3,4,5$), ($4,7,8$) and ($9,10,11$). The first two as well as the second and third triangles have a common edge (edge $3$ and edge $4$ respectively). Hence, these pairs of triangles are connected with an edge in the triangle-graph. The fourth triangle does not share any edges with the other triangles, therefore, it has no adjacent edges in the triangle-graph.}
  \label{fig:transformation}
\end{figure}
\begin{figure}[t]
  \centering
  \includegraphics[width=.8\linewidth]{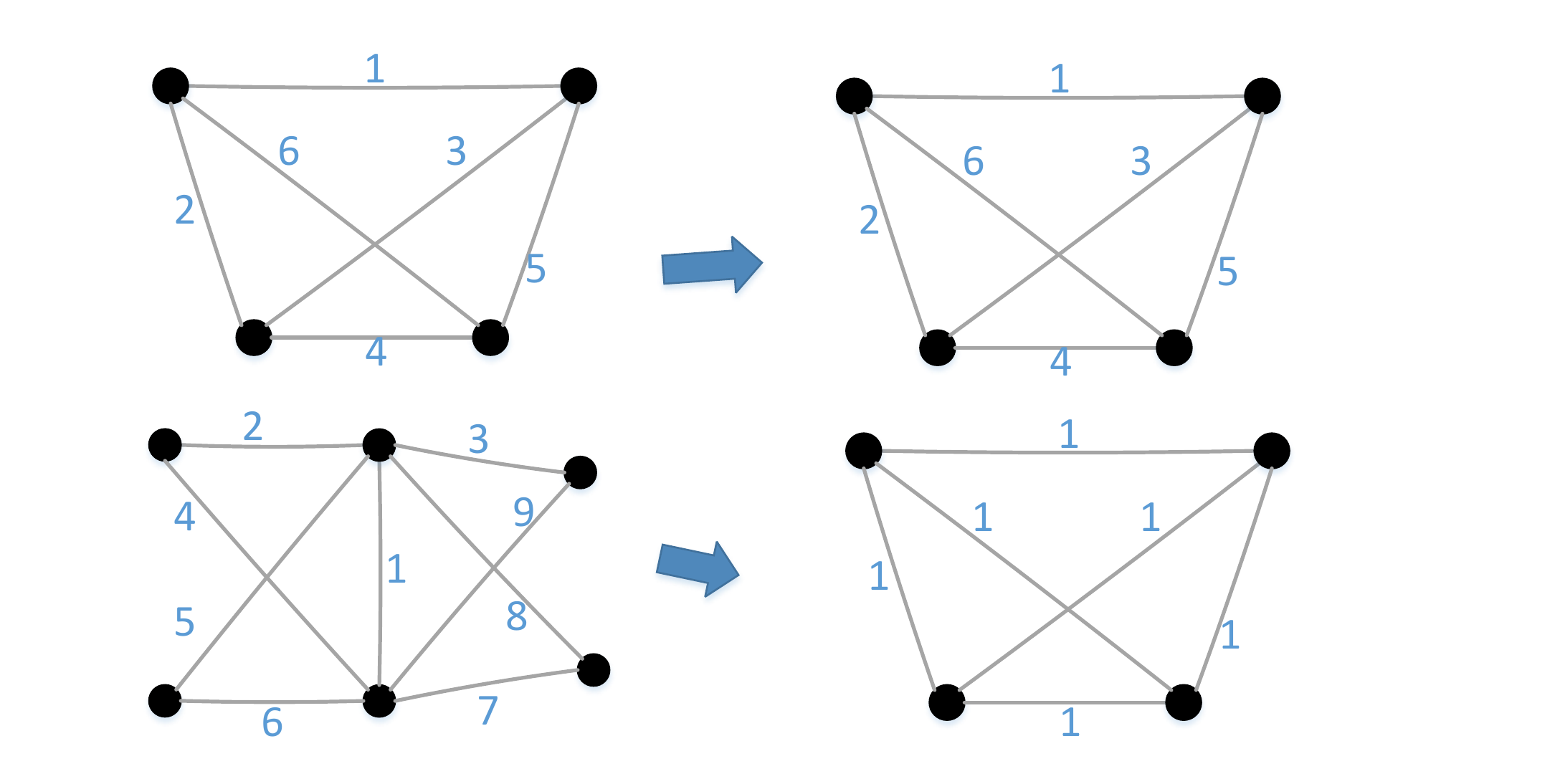}
  \caption{Two input graphs (left) and their triangle-graphs (right). The two triangle-graphs are structurally equivalent although the input graphs are not.}
  \label{fig:comparison}
\end{figure}
As can be seen from the Figure, the triangle-graphs emerging from the two input graphs are structurally equivalent, and hence, they have the same degree density.
As a result, if the two graphs are components of a larger graph and there are no other subgraphs with higher value, they are equally likely solutions to the DS problem.
However, it is obvious that the upper graph suits better our purpose, and we would like our algorithm to prefer this compared to the lower graph.

To account for this problem, we define a new density measure which we call the \textit{triangle-graph density}.
\begin{definition}[Triangle-Graph Density]
Given an undirected, unweighted graph $G = (V,E)$, first construct its triangle graph $G' = (V',E')$.
For any $S' \subseteq V'$, we define its triangle-graph density as $f(S') = \frac{d(S')}{|S'|}$ where $d(S') = \sum_{v \in S'} \min_{l \in L(v)} \big( deg_{S'}(v,l) \big)$, $L(v)$ the set of labels of the edges adjacent to $v$ (three labels at most), and $deg_{S'}(v,l)$ the number of edges that are adjacent to $v$ in the subgraph induced by $S'$ and are assigned the label $l$.
\end{definition}
The triangle-graph density will allow the discovery of subgraphs with high values of density $\delta$.
This is due to the fact that for each triangle $t$ in $G$, the function takes into account the number of neighbors from all three edges of $t$.
If a triangle $t$ corresponding to the vertex $v$ in $G'$ shares one of its edges with many other triangles, but the other two edges with no triangles, then $\min_{l \in L(v)} \big( deg_{S'}(v,l) \big) = 0$.
Therefore, even if $t$ has many neighbors, it contributes nothing to the triangle-graph density.
Triangle-graph density seeks for subgraphs whose vertices belong to edges which all consist of large sets of vertices.
Cliques are natural candidates for maximizing the function since all their edges are shared between several triangles.

We next introduce the \textit{triangle-graph densest subgraph} problem, the optimization problem we address in this paper.
\begin{problem}[TGDS problem]
Given an undirected, unweighted graph $G = (V,E)$, create its triangle-graph $G' = (V',E')$, and find a subset of vertices $S^* \subseteq V'$ such that $f(S^*) = \argmax_{S' \subseteq V'} f(S')$.
\end{problem}
After optimizing the triangle-graph density, we end up with a set of vertices $S' \subseteq V'$ and from these we obtain the set of vertices $S \subseteq V$ that corresponds to the resulting subgraph.
The set $S$ consists of all the vertices that form the triangles in $S'$.
It is clear that the TGDS problem can result in subgraphs with high values of density $\delta$.

What needs to be investigated next is what are the properties of the extracted subgraphs and how they differ from the ones extracted from existing methods.
The proposed \textit{triangle-graph densest subgraph} (TGDS) problem seems to be very related to the \textit{triangle densest subgraph} (TDS) problem introduced by Tsourakakis in \cite{tsourakakis2015k}.
However, as we will show next, the two problems can result in different solutions, and the subgraphs returned by TGDS are closer to being near-cliques compared to the ones returned by TDS.
Consider the graph $G$ and its triangle-graph $G'$ both shown in Figure~\ref{fig:example}.
\begin{figure}[t]
  \centering
  \includegraphics[trim = 10mm 5mm 1mm 5mm, clip, width=.8\linewidth]{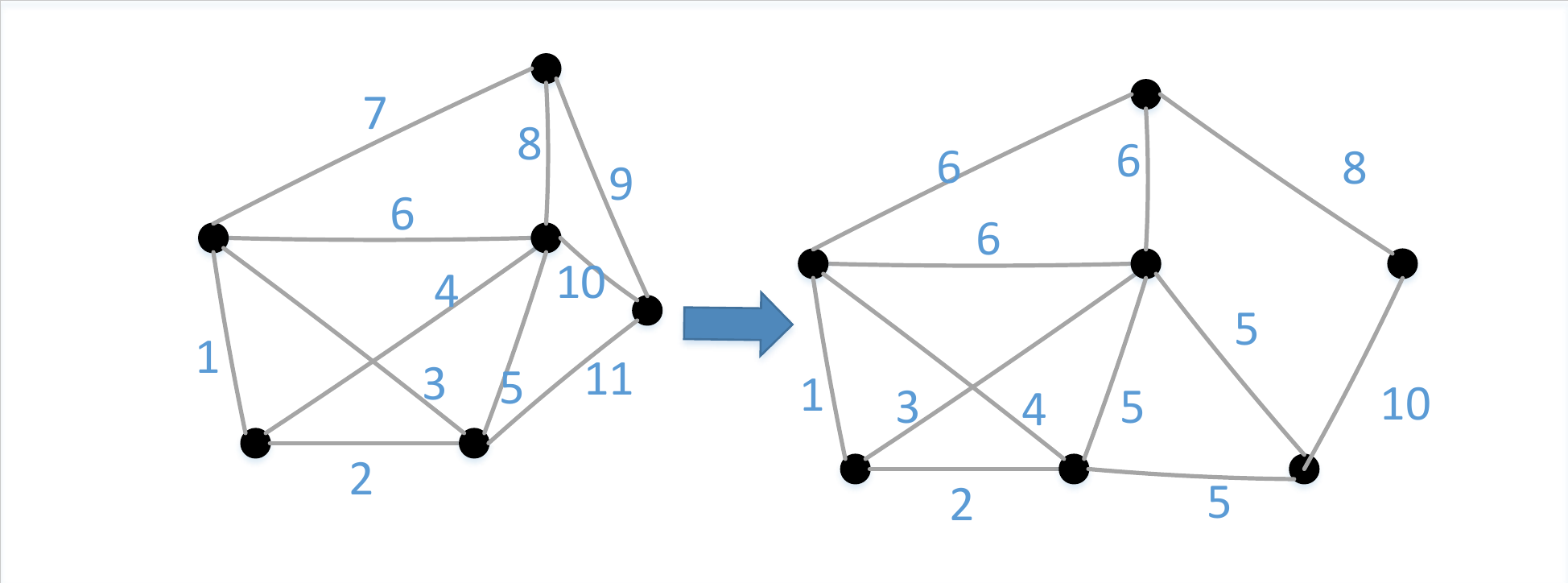}
  \caption{Example of an input graph (left)  and the triangle-graph (right) created from it. There are $7$ triangles in the input graph defined by the following triads of edges: ($1,2,3$), ($1,4,6$), ($2,4,5$), ($3,5,6$), ($6,7,8$), ($8,9,10$) and ($5,10,11$).}
  \label{fig:example}
\end{figure}
The optimal solution of TDS is the whole graph.
Conversely, the optimal solution of TGDS is the subgraph induced by the vertices that form the $4$-clique.
Hence, the optimal solution of the proposed problem is a clique, while the optimal solution of TDS is a larger graph with lower density $\delta$.
The above example demonstrates that the optimal solutions of TGDS correspond to subgraphs that exhibit a stronger near-clique structure compared to TDS.

The process of creating the $k$-clique graph for $k > 3$ is similar to the one described above for $k = 3$.
Specifically, to construct the $k$-clique graph $G'=(V',E')$, we first extract all the $k$-cliques from $G$.
Then for each $k$-clique in $G$, we create a vertex $v$ in $G'$.
Two vertices $v_1,v_2 \in V'$ are connected with an edge if the corresponding cliques share a common $(k-1)$-clique in $G$.
For example, for $k = 4$, if two $4$-cliques in $G$ share a common triangle, an edge is drawn between them in $G'$.
Each $(k-1)$-clique in $G$ is assigned a unique label and the edges of the $k$-clique graph are assigned the labels of the $(k-1)$-cliques that are shared between their two endpoints.
Then, the \textit{$k$-clique-graph density} and the \textit{$k$-clique-graph densest subgraph} ($k$-clique-GDS) problem are defined in a similar way as in the case of triangles.
The algorithms presented in the next Section for maximizing triangle-graph density can be generalized to maximizing the $k$-clique-graph density.
However, extracting $k$-cliques for $k > 3$ is a computationally demanding task, and hence, we restrict ourselves to the case where $k = 3$.

\section{Proposed Algorithm}\label{sec:proposed_methods}
\noindent
In this Section, we present a greedy algorithm for solving the TGDS problem.
The algorithm is inspired by previously-introduced algorithms in the field of dense subgraph discovery.
In what follows, we assume that we have extracted all triangles from the input graph and we have created the triangle-graph.
Note that, for simplicity of notation, from now on, we denote by $G = (V,E)$ the triangle-graph and not the input graph.
We also denote by $q_S(v)$ the minimum degree of vertex $v$ with respect to the three labels of its adjecent edges in the subgraph induced by $S$, that is $q_S(v) = \min_{l \in L(v)} \big( deg_{S}(v,l) \big)$.

We next provide an efficient algorithm for extracting a set of vertices $S \subseteq V$ with high value of triangle-graph density $f(S)$.
The proposed algorithm is an adaptation of the greedy algorithm of Asahiro et al. \cite{asahiro2000greedily}.
The algorithm is illustrated as Algorithm~\ref{alg:greedy}.
\begin{algorithm}[t]
\caption{Greedy algorithm}
\label{alg:greedy}
\begin{algorithmic}
\Require{graph $G = (V,E)$}
\Ensure{Subset of vertices $S \subseteq V$}
\State{$S_{|V|} \gets V$}
\For{$i \gets |V|$ to $1$}
  \State{Let $v$ be the vertex whose minimum value of the three degrees is the smallest in the subgraph induced by $S_i$}
  \State{$S_{i-1} \gets S_i \setminus \{v\}$}
\EndFor
\State{$S \gets \argmax_{i=1,\ldots,|V|} f(S_i)$}
\end{algorithmic}
\end{algorithm}
The algorithm iteratively removes the vertex $v$ whose value $d(v)$ is the smallest among all vertices.
Subsequently, it computes the triangle-graph density of the subgraph induced by the remaining vertices.
The output is the subgraph over all the produced subgraphs that maximizes triangle-graph density.
The algorithm is linear to the number of vertices and the number of edges of the triangle-graph, hence its complexity is $\mathcal{O}(t+y)$ where $t$ is the number of triangles in the input graph and $y$ is the number of edges of the triangle-graph.

\begin{theorem}
Let $S$ be the set of vertices returned after the execution of Algorithm~\ref{alg:greedy} and let $S^*$ be the set of vertices of the optimal subgraph.
Consider the iteration of the greedy algorithm just before the first vertex $u$ that belongs in the optimal set $S^*$ is removed, and let $S_I$ denote the vertex set currently kept in that iteration.
Let also $q_{S_I}(u)$ be the minimum degree of vertex $u$ in $S_I$ with respect to the three labels of its adjacent edges.
Then, it holds that
\begin{equation*}
	f(S) \geq \frac{|S^*|}{|S_I|}f_G^* + \bigg( 1-\frac{|S^*|}{|S_I|} \bigg) q_{S_I}(u)
\end{equation*}
\end{theorem}
\begin{proof}
Given a subset of vertices $S \subseteq V$ and a vertex $v$, let $q_S(v)$ be the minimum degree of vertex $v$ with respect to its three labels.
Let also $S^*$ be the vertices of the optimal subgraph.
The optimal value of the function is obtained for the set of vertices $S^*$ and is equal to $f(S^*) = \nicefrac{d(S^*)}{|S^*|}$.

Consider the iteration of the greedy algorithm just before the first vertex $u$ that belongs in the optimal set $S^*$ is removed.
Let $S_I$ denote the set of vertices still present before the removal of $u$.
The value of the function for the set of vertices $S_I$ is then $f(S_I) = \nicefrac{d(S_I)}{|S_I|}$.

Since $S^* \subseteq S_I \subseteq V$, it holds that $q_V(v) \geq q_{S_I}(v) \geq q_{S^*}(v)$, $\forall v$.
In each iteration, the algorithm removes the vertex with the minimum degree with respect to the three labels of its adjacent edges.
Since $u$ is the first vertex to be removed by the algorithm, it is also easy to see that $q_{S_I}(v) \geq q_{S_I}(u) \geq q_{S^*}(u)$.
Therefore,
\begin{equation*}
\begin{split}
  & f(S_I) = \frac{d(S_I)}{|S_I|} \\
  &= \frac{\sum\limits_{v \in S^*} q_{S^*}(v) + \sum\limits_{v \in S^*} \big( q_{S_I}(v) - q_{S^*}(v) \big) + \sum\limits_{v \in S_I \backslash S^*} q_{S_I}(v)}{|S_I|} \\
  & \geq \frac{\sum\limits_{v \in S^*} q_{S^*}(v) + \sum\limits_{v \in S_I \backslash S^*} q_{S_I}(v)}{|S_I|} \\
  & = \frac{|S^*|f(S^*) + \sum\limits_{v \in S_I \backslash S^*} q_{S_I}(v)}{|S_I|} \geq \frac{|S^*|f(S^*) + (|S_I| - |S^*|) q_{S_I}(u)}{|S_I|} \\
  & = \frac{|S^*|}{|S_I|}f(S^*) + \bigg( 1-\frac{|S^*|}{|S_I|} \bigg) q_{S_I}(u)
\end{split}
\end{equation*}
The algorithm returns a set of nodes $S$ which is the best over all iterations, hence we obtain
\begin{equation*}
  f(S) \geq f(S_I) \geq \frac{|S^*|}{|S_I|}f(S^*) + \bigg( 1-\frac{|S^*|}{|S_I|} \bigg) q_{S_I}(u)
\end{equation*}
\end{proof}

From the above result, we can see that the bound provided by the approximation algorithm highly depends on the relationship between $|S_I|$, the size of the vertex set just before the first vertex of $S^*$ is removed, and $|S^*|$, the size of the optimal set.
It also depends on the relationship between the optimal value of the triangle-graph density $f(S^*)$ and the minimum degree $q_{S_I}(u)$ of the first vertex of the optimal set $S^*$ to be removed from $S_I$ with respect to its three labels.
The difference between $|S_I|$ and $|S^*|$, and between $f(S^*)$ and $q_{S_I}(u)$ is not very large in practice, and the algorithm leads to subgraphs with quality almost equal to that of the optimal subgraphs.

\section{Experiments and Evaluation}\label{sec:experiments}
\noindent
In this Section, we present the evaluation of the proposed approach for extracting dense subgraphs.
We first give details about the datasets that we used for our experiments.
We then present the employed experimental settings.
And we last report on the results obtained by our approach and some other methods.

\subsection{Experimental Setup}
\noindent
For the evaluation of the proposed algorithms, we employed several publicly available graphs.
The algorithms are applicable to simple unweighted, undirected graphs.
Hence, we made all graphs simple by ignoring the edge direction in the case of directed graphs and by removing self-loops and egde weights, if any.
Table~\ref{tab:graphs} shows statistics of these graphs.
\begin{table}[t]
\centering
\caption{Graphs used for evaluating the algorithms.}
\label{tab:graphs}
\scriptsize
\def\arraystretch{1.2}
\begin{tabular}{|l|r|r|} \hline
\multicolumn{1}{|c|}{\textbf{Graph}} & \multicolumn{1}{c|}{$\mathbf{|V|}$} & \multicolumn{1}{c|}{$\mathbf{|E|}$} \\ \hline
Karate & 34 & 78 \\ 
Dolphins & 62 & 159 \\ 
Lesmis & 77 & 254 \\
Adjnoun & 112 & 425 \\
Football & 115 & 613 \\ 
Polbooks & 105 & 441 \\
Celegansneural & 297 & 2,148 \\
Polblogs & 1,224 & 16,715 \\
Power & 4,941 & 6,594 \\
Wiki-Vote & 7,115 & 100,762 \\
ca-CondMat & 23,133 & 93,439 \\
p2p-Gnutella31 & 62,586 & 147,892 \\
Slashdot0902 & 82,168 & 504,230 \\
email-EuAll & 265,009 & 364,481 \\
web-NotreDame & 325,729 & 1,497,134 \\
Amazon & 334,863 & 925,872 \\
Youtube & 1,134,890 & 2,987,624 \\
roadNet-CA & 1,965,206 & 2,766,607 \\
\hline
\end{tabular}
\end{table}
The first ten datasets were obtained from UCIrvine Network Data Repository\footnote{\url{https://networkdata.ics.uci.edu/index.php}}, while the remaining datasets were obtained from Stanford SNAP Repository\footnote{\url{http://snap.stanford.edu/data/index.html}}.
We compared the proposed algorithms with algorithms that solve the \textit{densest subgraph} (DS), the \textit{triangle densest subgraph} (TDS) and the \textit{optimal quasi-clique} (OQC) problems.
For the first two (DS and TDS problems), there are algorithms that solve these problems exactly in polynomial time.
Hence, for small-sized datasets, we present the results obtained from both the exact and greedy approximation algorithms for each problem.
For larger datasets, we report only on the results achieved by the greedy approximation algorithms.
With regards to the objective function of the OQC problem, we set the value of parameter $\alpha$ equal to $\nicefrac{1}{3}$ as suggested in \cite{tsourakakis2013denser}.
All algorithms were implemented in Python\footnote{Code is available at \url{https://github.com/giannisnik/k-clique-graphs-dense-subgraphs}} and all experiments were conducted on a single machine with a $3.4$GHz Intel Core i$7$ processor and $32$GB of RAM.
To assess the quality of the extracted subgraphs, we employed the following measures: the density of the extracted subgraph $\delta(S) = |E(S)|/\binom{|S|}{2}$, the density with respect to the number of triangles $\tau(S) = t(S)/\binom{|S|}{3}$, that is the number of triangles in $S$ over the total possible triangles, and the size of the subgraph $|S|$.
The $\delta$ and $\tau$ measures take values between $0$ and $1$.
The larger their value, the closer the subgraph to being a clique.
Therefore, we are interested in finding large subgraphs (large value of $|S|$) with $\delta$ and $\tau$ values close to $1$.
 
\subsection{Results and Discussion}
\noindent
Table~\ref{tab:small_results}
\begin{table}[t]
\centering
\caption{Comparison of the extracted subgraphs by Goldberg's exact algorithm for the DS problem (Exact DS), Charikar's $\frac{1}{2}$ approximation algorithm for the DS problem (Greedy DS), Tsourakakis's algorithm for the TDS problem (Exact TDS), Tsourakakis's $\frac{1}{3}$ approximation algorithm for the TDS problem (Greedy TDS), Tsourakakis et al.'s greedy approximation algorithm for the OQC problem (Greedy OQC), and our greedy approximation algorithm for the TGDS problem (Greedy TGDS).}
\label{tab:small_results}
\scriptsize
\def\arraystretch{1.2}
\resizebox{\textwidth}{!} {
\begin{tabular}{|l|ccc|ccc|ccc|ccc|ccc|ccc|}
\hline
\multirow{2}{*}{\textbf{Dataset}} & \multicolumn{3}{c|}{\textbf{Exact DS}} & \multicolumn{3}{c|}{\textbf{Greedy DS}} & \multicolumn{3}{c|}{\textbf{Exact TDS}} & \multicolumn{3}{c|}{\textbf{Greedy TDS}} & \multicolumn{3}{c|}{\textbf{Greedy OQC}} & \multicolumn{3}{c|}{\textbf{Greedy TGDS}}\\
& $|S|$ & $\delta$ & $\tau$ & $|S|$ & $\delta$ & $\tau$ & $|S|$ & $\delta$ & $\tau$ & $|S|$ & $\delta$ & $\tau$ & $|S|$ & $\delta$ & $\tau$ & $|S|$ & $\delta$ & $\tau$ \\ \hline
Karate & 16 & 0.35 & 0.05 & 16 & 0.35 & 0.05 & 6 & 0.93 & 0.80 & 6 & 0.93 & 0.80 & 10 & 0.55 & 0.18 & 6 & 0.93 & 0.80 \\
Dolphins & 20 & 0.32 & 0.04 & 36 & 0.17 & 0.01 & 7 & 0.80 & 0.54 & 6 & 0.93 & 0.80 & 13 & 0.47 & 0.11 & 6 & 0.93 & 0.80 \\
Lesmis & 23 & 0.49 & 0.18 & 23 & 0.49 & 0.18 & 13 & 0.88 & 0.71 & 13 & 0.88 & 0.71 & 22 & 0.50 & 0.19 & 12 & 0.93 & 0.83 \\
Adjnoun & 48 & 0.20 & 0.01 & 44 & 0.22 & 0.01 & 41 & 0.23 & 0.01 & 41 & 0.23 & 0.01 & 16 & 0.48 & 0.11 & 7 & 0.85 & 0.62 \\
Football & 115 & 0.09 & 0.00 & 115 & 0.09 & 0.00 & 18 & 0.48 & 0.20 & 18 & 0.48 & 0.20 & 10 & 0.88 & 0.66 & 18 & 0.48 & 0.20 \\
Polbooks & 24 & 0.41 & 0.09 & 48 & 0.19 & 0.02 & 20 & 0.49 & 0.15 & 36 & 0.26 & 0.04 & 14 & 0.67 & 0.30 & 13 & 0.69 & 0.34 \\
\hline
\end{tabular}
}
\end{table}
summarizes the results obtained on small-sized graphs.
We observe that on the small-sized graphs, the proposed algorithm (Greedy TGDS) returns in general subgraphs that are closer to being a clique compared to the competing algorithms.
As we can see from the Table, the densities $\delta$ and $\tau$ of the subgraphs extracted by our algorithm are relatively high.
Our initial intention was to design an algorithm for finding a set of vertices with many edges between them.
The obtained results verify our intuition that the proposed approach is capable of finding near-cliques.
Furthermore, we show in Table~\ref{tab:bound}
\begin{table}[t]
\centering
\caption{Triangle-graph densities of the subgraphs extracted by an exact agorithm and the proposed greedy approximation algorithm.}
\label{tab:bound}
\def\arraystretch{1.1}
\scriptsize
\begin{tabular}{|l|c|c|}
\hline
\textbf{Dataset} & \textbf{Exact TGDS} & \textbf{Greedy TGDS} \\ \hline
Karate & 2.25 & 2.25 \\
Dolphins & 2.25 & 2.25 \\
Lesmis & 7.60 & 7.60 \\
Adjnoun & 2.39 & 2.36 \\
Football & 6.0 & 6.0 \\
Polbooks & 4.02 & 3.89 \\
\hline
\end{tabular}
\end{table}
the triangle-graph density of the subgraphs extracted by a brute-force exact algorithm and the proposed greedy approximation algorithm.
We notice that on four out of the six graphs, the two densities are equal to each other, while on the other two, they are very close to each other.
The obtained results indicate that the greedy algorithm achieves approximation ratios close to $1$ on real-world networks.
Hence, the approximation algorithm is nearly-optimal in practice.

Next, we present results obtained on larger graphs.
Specifically, Table~\ref{tab:large_results} 
\begin{table}[t]
\centering
\caption{Comparison of the extracted subgraphs by Charikar's $\frac{1}{2}$ approximation algorithm for the DS problem (Greedy DS), Tsourakakis's $\frac{1}{3}$ approximation algorithm for the TDS problem (Greedy TDS), Tsourakakis et al.'s greedy approximation algorithm for the OQC problem (Greedy OQC), and our greedy approximation algorithm for the TGDS problem (Greedy TGDS).}
\label{tab:large_results}
\def\arraystretch{1.2}
\scriptsize
\begin{tabular}{|l|ccc|ccc|ccc|ccc|}
\hline
\multirow{2}{*}{\bf{Dataset}} & \multicolumn{3}{c|}{\textbf{Greedy DS}} & \multicolumn{3}{c|}{\textbf{Greedy TDS}} & \multicolumn{3}{c|}{\textbf{Greedy OQC}} & \multicolumn{3}{c|}{\textbf{Greedy TGDS}}\\
& $|S|$ & $\delta$ & $\tau$ & $|S|$ & $\delta$ & $\tau$ & $|S|$ & $\delta$ & $\tau$ & $|S|$ & $\delta$ & $\tau$ \\
\hline
Celegansneural & 127 & 0.13 & 0.005 & 30 & 0.47 & 0.13 & 22 & 0.61 & 0.25 & 24 & 0.55 & 0.21 \\
Polblogs & 278 & 0.20 & 0.020 & 102 & 0.54 & 0.195 & 100 & 0.55 & 0.202 & 74 & 0.67 & 0.343 \\
Power & 31 & 0.20 & 0.021 & 12 & 0.54 & 0.195 & 12 & 0.54 & 0.195 & 12 & 0.54 & 0.195 \\
Wiki-Vote & 828 & 0.11 & 0.004 & 464 & 0.19 & 0.014 & 133 & 0.47 & 0.131 & 152 & 0.42 & 0.104 \\
ca-CondMat & 26 & 1.0 & 1.0 & 26 & 1.0 & 1.0 & 26 & 1.0 & 1.0 & 26 & 1.0 & 1.0 \\
p2p-Gnutella31 & 1,549 & 0.005 & 0.0 & 10 & 0.40 & 0.11 & 14 & 0.48 & 0.0 & 22 & 0.15 & 0.016 \\
soc-Slashdot0902 & 219 & 0.39 & 0.097 & 171 & 0.50 & 0.165 & 155 & 0.54 & 0.200 & 145 & 0.56 & 0.225 \\
email-EuAll & 505 & 0.13 & 0.005 & 200 & 0.29 & 0.041 & 97 & 0.51 & 0.164 & 91 & 0.52 & 0.179 \\
web-NotreDame & 1,367 & 0.11 & 0.012 & 457 & 0.34 & 0.114 & 305 & 0.51 & 0.255 & 155 & 1.0 & 1.0 \\
Amazon & 9 & 0.91 & 0.761 & 16 & 0.45 & 0.178 & 9 & 0.91 & 0.761 & 170 & 0.03 & 0.001 \\
Youtube & 1,860 & 0.049 & 0.0006 & 729 & 0.11 & 0.005 & 125 & 0.46 & 0.115 & 442 & 0.17 & 0.012 \\
roadNet-CA & 19,899 & 0.0001 & 0.0 & 168 & 0.017 & 0.0002 & 5 & 0.80 & 0.40 & 168 & 0.017 & 0.0002 \\
\hline
\end{tabular}
\end{table}
compares the four approaches on $12$ graphs.
In general, the proposed algorithm still manages to extract subgraphs with high values of $\delta$ and $\tau$.
However, on two graphs (Amazon, roadNet-CA), it fails to discover high-quality subgraphs in terms of density.
Overall, the Greedy DS algorithm returns the largest subgraphs, followed by the Greedy TDS algorithm, while the Greedy OQC algorithm and the proposed algorithm return smaller subgraphs with higher values of density.
We notice that the subgraphs extracted by the proposed greedy approximation algorithm resemble most those extracted by the Greedy TDS algorithm.
On the ca-CondMat dataset, all the algorithms extract the same subgraph.
There is a large clique hidden in this graph and all the algorithms manage to find it.

\section{Application}\label{sec:application}
\noindent
In this Section, we apply the proposed algorithm to a central problem in Natural Language Processing: extracting keywords from a textual document.
Keyword extraction finds applications in several fields from information retrieval to text classification and summarization.
Given a document $d$, we can represent it as a statistical \textit{graph-of-words}, following earlier approaches in keyword extraction \cite{mihalcea_textrank:_2004,rousseau2015main,tixiergraph} and in summarization \cite{meladianos2015degeneracy}.
The construction of the graph is preceded by a preprocessing phase where standard text processing tasks are performed.
The processed document is then transformed into an unweighted, undirected graph $G$ whose vertices represent unique terms and whose edges represent co-occurrences between the connected terms within a fixed-size window.
We then employ the proposed algorithm to extract a dense subgraph from $G$.
The vertices of the subgraph act as representative keywords of the document.

To demonstrate the ability of the proposed approach to identify meaningful keywords, we extracted the text of this paper and we transformed it into a graph $G$ using a window of size $3$ (each word is connected with an edge with each one of its two preceding and two following words, if any).
We then extracted a dense subgraph from $G$ using the proposed greedy approximation algorithm.
The output subgraph consists of the following $31$ vertices:
\begin{center}
\texttt{subgraphs, labels, maximizes, vertices, k, first, cliques, triangle, subgraph, algorithm, triangles, value, optimal, density, edges, large, g, number, vertex, given, function, clique, graph, v, e, set, problem, input, edge, extract, hence}
\end{center}
As we can observe, the extracted keywords capture the main concepts of the paper.

\section{Conclusion}\label{sec:conclusion}
\noindent
In this paper, we propose a novel approach for extracting dense subgraphs.
Given a graph, our algorithm first transforms it to a $k$-clique-graph.
We then introduce a simple density measure to extract high-quality subgraphs.
We propose a greedy approximation algorithm for maximizing the density function. 
We evaluate our proposed approach for the case where $k=3$ on real graphs and we compare it with other popular measures.
We also evaluate our proposed method on the task of keyword extraction from textual documents.
Overall, our algorithm shows good performance in finding large near-cliques, and can serve as a useful addition to the list of dense subgraph discovery algorithms.

\bibliographystyle{splncs03}
\bibliography{biblio.bib}

\end{document}